\documentclass[runningheads]{llncs}

\usepackage[utf8]{inputenc}
\usepackage[T1]{fontenc}
\usepackage{microtype,verbatim,dsfont}
\usepackage{graphicx}
\graphicspath{{figures/}}
\usepackage{subcaption}
\usepackage{cite}
\usepackage[basic]{complexity}
\usepackage{todonotes}
\usepackage{amsmath}
\usepackage{amssymb}
\usepackage{enumerate}
\usepackage{placeins}
\usepackage[pdfpagelabels,colorlinks,allcolors=blue]{hyperref}
\usepackage[linesnumbered,lined,boxed,commentsnumbered]{algorithm2e}

\newtheorem{observation}{Observation}
\newtheorem{idea}{Idea}

\title{\texttt{CelticGraph}: Drawing Graphs as Celtic Knots and Links}

\author{
Peter Eades\inst{1} \and 
Niklas Gröne\inst{2} \and 
Karsten Klein\inst{2}\orcidID{0000-0002-8345-5806} \and 
Patrick Eades\inst{5} \and 
Leo Schreiber\inst{4} \and
Ulf Hailer\inst{2} \and
Falk Schreiber\inst{2,3} }

\authorrunning{Peter Eades et al.}
\institute{
University of Sydney, Australia \and 
University of Konstanz, Germany \and 
Monash University, Australia \and 
Independent Scholar, Germany \and
Melbourne University, Australia}

\begin{document}

\maketitle

\begin{figure}
  \centering
  \vspace*{-1ex}
  \includegraphics[width=0.85\textwidth]{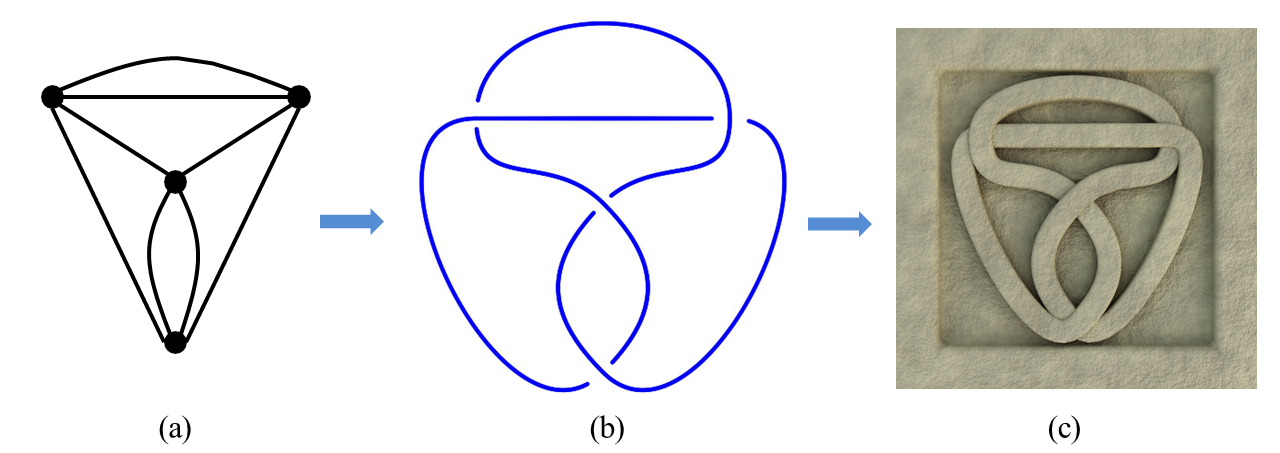}
  \caption{\emph{Using the {\tt CelticGraph} framework, we take the graph $K'_4$ (a) ($K_4$ with some duplicate edges), and create a knot drawing of $K'_4$ (b); from this we render the graph as a Celtic Knot with a sandstone texture (c).}
  }
  \label{fi:teaserV2}
  \vspace*{-2ex}
\end{figure}

\begin{abstract}
    Celtic knots are an ancient art form often attributed to Celtic cultures, used to decorate monuments and manuscripts, and to symbolise eternity and interconnectedness.
    This paper describes the framework \texttt{CelticGraph} to draw graphs as Celtic knots and links.
    The drawing process raises interesting combinatorial concepts in the theory of circuits in planar graphs. Further, \texttt{CelticGraph} uses a novel algorithm to represent edges as B\'ezier curves, aiming to show each link as a smooth curve with limited curvature.
\end{abstract}

\section{Introduction}


Celtic knots are an ancient art form often attributed to Celtic cultures. These elaborate designs (also called ``endless knots'') were used to decorate monuments and manuscripts, and they were often used to symbolise eternity and interconnectedness. Celtic knots are a well-known visual representation made up of a variety of interlaced knots, lines, and stylised graphical representations. The patterns often form continuous loops with no beginning or end (knot) or a set of such loops (links). In this paper we will use the Celtic knot visualisation metaphor to represent specific graphs in the form of ``knot diagrams''. 

We show how to draw a 4-regular planar graph\footnote{Graphs used in this paper can contain multiple edges and loops (also called pseudographs or multigraphs).}  as a knot/link diagram. This involves constructing certain circuits\footnote{For the formal definition of a circuit see Sect.~\ref{sec:c}} in the 4-regular planar graph. Further, we show how to route graph edges so that the underlying links are aesthetically pleasing. This involves some optimisation problems for cubic B\'ezier curves. We also provide an implementation of the presented methods as an add-on for  \texttt{Vanted}~\cite{JunkerKS06b}. This  system allows to transform a graph into a knot (link) representation and to interactively change the layout of both, the graph as well as the knot. In addition, knots can be exported to the 3D renderer Blender~\cite{Blender} to allow for artistic 3D renderings of the knot. Figure~\ref{fi:teaserV2} shows a 4-regular planar graph, its knot representation and a rendering of the knot.


\section{Background}

This paper has its roots in three disciplines: Mathematical knot theory, Celtic cultural history, and graph drawing. We briefly review the relevant parts of these diverse fields in Sect.~\ref{se:knots}, \ref{se:celtic}, and \ref{se:gd}. Further, in Sect.~\ref{se:B\'ezier} we review relevant properties of B\'ezier curves, which are a key ingredient to \texttt{CelticGraph}.

\subsection{Knot theory}
\label{se:knots}

The mathematical theory of knots and links investigates interlacing curves in three dimensions; this theory has a long and distinguished history in Mathematics~\cite{james}. The motivating problem of Knot Theory is \emph{equivalence}: whether two knots can be deformed into each other. A common technique involves projecting the given curves from three dimensions into the plane; the resulting ``knot diagram'' is a 4-regular planar graph, with vertices at the points where the curve crosses itself (in the projection). For example, a picture of the \emph{trefoil knot} and its knot diagram are in Fig.~\ref{fi:trefoilDiagram}.
Properties of a knot or link may be deduced from the knot diagram, and the equivalence problem can sometimes be solved using knot diagrams.

\begin{figure}[t]
  \centering      
  \vspace*{-1ex}
  \includegraphics[width=5.5cm]{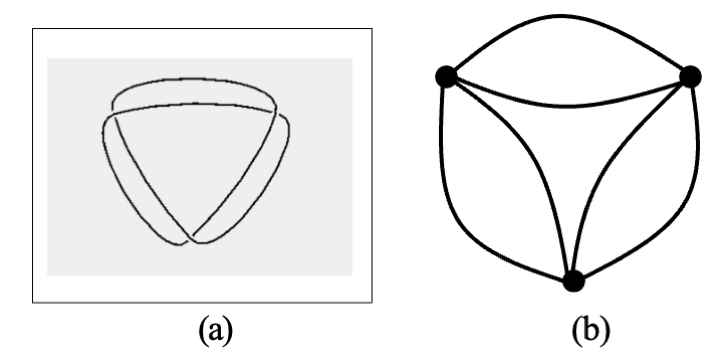}
  \caption{\emph{(a) The {\emph{trefoil knot}} and (b) its resulting ``knot diagram''.}}
  \label{fi:trefoilDiagram}
  \vspace*{-2ex}
\end{figure}

\subsection{Celtic art}
\label{se:celtic}

Knot patterns (``Celtic knots'') are often described as a characteristic ornament of so-called ``Celtic art''. In fact, since the epoch of the Waldalgesheim style (4$^\text{th}$/3$^\text{rd}$ century BC), Celtic art (resp.~ornamentation) is characterised by complex, often geometric patterns of interlinked, opposing or interwoven discs, loops and spirals. The floral models originate from Mediterranean art; in the Celtic context, they were deconstructed, abstracted, arranged paratactically or intertwined~\cite{M12,RB01-197}. 
It is still unclear whether the import of Mediterranean ornamental models was accompanied by the adoption of their meaning. However, the selective reception of only certain motifs suggests rather an adaptation based on specific Celtic ideas, which we cannot reconstruct exactly due to a lack of written sources.

In today's popular understanding, a special role in the transmission of actual or supposed ``Celtic''  art is attributed to the early medieval art of Ireland~\cite{Ma15-158,Roeb12-460}. However, such a restriction of Irish or insular art to exclusively Celtic origins would ignore the historical development of the insular-Celtic context in Ireland and the British Isles. The early medieval art of Ireland is partially rooted in indigenous Celtic traditions, but was also shaped by Late Antique Roman, Germanic and Anglo-Saxon, Viking and  Mediterranean-Oriental models~\cite{FN02-138}. 
The knot and tendril patterns of the 7$^\text{th}$/8$^\text{th}$ century can also be traced back to Mediterranean-Oriental manuscripts. Such patterns were subsequently used in Anglo-Saxon art, transmitted by braided ribbon ornaments and other patterns in the Germanic ``Tierstil'' (e.\,g.~on Late Antique soldiers' belts). For example, the famous Tara Brooch created in Ireland in the late 7$^\text{th}$ or early 8$^\text{th}$ century features both native and Germanic ~a combination of corresponding motifs~\cite{Roeb12-512}.  Also the knot patterns and braided/spiral ornaments described as typical ``Celtic'' such as in the Book of Kells~\cite{bookOfKells}  and other manuscripts 
can be linked to Germanic/Anglo-Saxon and late Roman traditions. The ornamentation today often perceived as ``Celtic'' is therefore less exclusive or typical ``Celtic'', but rather a result of diverse influences that reflect an equally complex historical-political development~\cite{Ma15-159}.
So it is not surprising that the so-called Celtic motifs often presented in tattoo studios of the 21st century, such as braided bands, are not of Celtic but Germanic origin~\cite{M12-119}.

Note that while ``Celtic knots''  are related to the mathematical theory of knots, the prime motivation of the two topics is different. For example, the \emph{Bowen knot}~\cite{Clark1827}, a commonly used decorative knot that appears in Celtic cultures, is uninteresting in the mathematical sense (it is clearly an ``unknot'').

\subsection{Graph drawing as art}
\label{se:gd}

Note that the purpose of \texttt{CelticGraph} is different from most Graph Drawing systems. Our aim is to produce decorative and artistically pleasing pictures of graphs, not to make pictures of graphs that effectively convey information and insight into data sets. Other examples of such kind of graph drawing approaches include the system by Devroye and Kruszewski to  render images of botanical trees based on random binary trees~\cite{DBLP:conf/gd/DevroyeK95},  a system \texttt{GDot-i} for drawing graphs as dot paintings inspired by the dot painting style of Central Australia~\cite{GDot,GDotPoster}, and research on bobbin lacework~\cite{bobbinLace}. Also related to our work are Lombardi graph drawings,  artistic representations of graphs that 
contain edges represented as circular arcs and vertices represented with perfect angular resolution~\cite{DuncanEGKN10a}.

\subsection{B\'ezier curves}
\label{se:B\'ezier}

The Gestalt law of continuity~\cite{koffka} implies that humans are more likely to follow continuous and smooth lines
rather than broken or jagged lines. To draw graphs as Celtic knots, certain circuits in the graph need to be drawn as smooth curves.

Computer Graphics has developed many models for smooth curves; one of the simplest is a \emph{B\'ezier curve}. A \emph{cubic B\'ezier curve} with control points $p_0, p_1, p_2, p_3$ is defined parametrically by:
\begin{equation}
\label{eq:B\'ezier}
    p(t) = (1-t)^3 p_0 + 3(1-t)^2 t p_1 + 3 (1-t) t^2 p_2 + t^3 p_3,
\end{equation}
for $0 \leq t \leq 1$.
The following properties of cubic B\'ezier curves are well-known~\cite{DBLP:books/daglib/0067131}:
\begin{itemize}
    \item The endpoints of the curve are the first and last control points, that is, $p(0) = p_0$ and $p(1) = p_3$.
    \item Every point on the curve lies within the convex hull of its control points.
    \item The line segments $(p_0,p_1)$ and $(p_3,p_2)$ are tangent to the curve at $p_0$ and $p_3$ respectively.
    We say $(p_0,p_1)$ and $(p_3,p_2)$ are  \emph{control tangents} of the curve.
    \item The curve is \emph{$C^k$ smooth} for all $k>0$, that is, all the derivatives are continuous.
\end{itemize}

Drawing each edge of a graph as a cubic B\'ezier curve ensures smoothness in the edges, and can improve readability~\cite{KaiXu}. However, for \texttt{CelticGraph} we need certain \emph{circuits} in the graph to be smooth curves, so we need the curves representing certain incident edges to be \emph{joined smoothly}. Suppose that $p(t)$ and $q(t)$ are two cubic B\'ezier curves that meet at a common endpoint. Then the curve formed by joining $p(t)$ and $q(t)$ is $C^1$ smooth as long as the control tangents to each curve at the common endpoint form a straight line; see Fig.\ref{fi:join}.
 
\begin{figure}[t]
  \centering
  \vspace*{-3ex}
  \includegraphics[width=4.5cm]{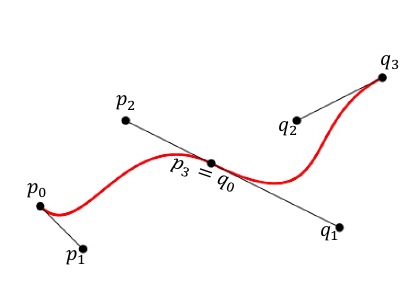}
     \vspace*{-2ex}
  \caption{\emph{Two cubic B\'ezier curves with control points $p_0, p_1, p_2, p_3$ and $q_0, q_1, q_2, q_3$, meeting at the point $p_3 = q_0$. The control tangents are shown in black; note that the points $p_2, p_3, q_1$ lie on a straight line and the join is $C_1$ and visually smooth.}}
  \label{fi:join}
   \vspace*{-2ex}
\end{figure}

Mathematically, $C^1$ smoothness is adequate. However, the infinitesimality of Mathematics sometimes does not model human perception well. For example, the curve in Fig.~\ref{fi:kink2} is mathematically smooth, but given a fixed-resolution screen and the limits of human perception, it appears to have a non-differentiable ``kink''.

\begin{figure}[t]
  \centering
  \vspace*{-1ex}
  \includegraphics[width=2.5cm]{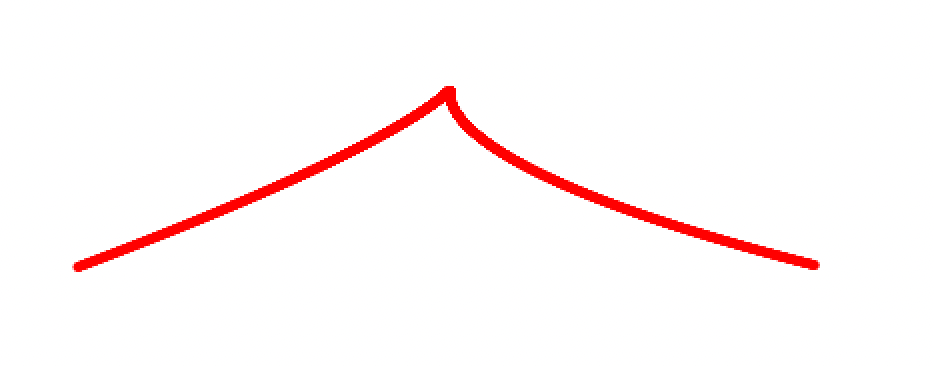}
  \caption{\emph{A cubic B\'ezier curve with a ``kink'', i.\,e.~a point of large curvature, near the middle. The curve is $C^1$ smooth; but the kink, together with the limits of human perception and screen resolution, mean that the curve does not look smooth.}}
  \label{fi:kink2}
  \vspace*{-2ex}  
\end{figure}

For this reason,  it is desirable that the \emph{curvature}~\cite{Ferguson} of each edge is not too large. Informally, the curvature $\kappa(t)$ is the ``sharpness'' of the curve. More formally, $\kappa(t)$ is the inverse of the radius of the largest circle that can sit on the curve at $p(t)$ without crossing the curve.
For a cubic B\'ezier curve $p(t) = (x(t), y(t))$, the curvature at $p(t)$ is given by~\cite{DBLP:books/daglib/0067131}:
\begin{equation}
    \label{eq:curvature}
    \kappa(t) = \frac{\left | \dot{x} \ddot{y} - \ddot{x} \dot{y} \right |}{\left ( \dot{x}^2 + \dot{y}^2 \right ) ^ {1.5}},
\end{equation}
where $\dot{f}$ denotes the derivative of $f$ with respect to $t$.
Note that $\kappa(t)$ is continuous except for values of $t$ where both $\dot{x}(t)$ and $\dot{y}(t)$ are zero.
For \texttt{CelticGraph}, we need $C^1$ smooth curves with reasonably small curvature.

\subsection{Related Work}

\subsubsection{Knot diagrams as Lombardi graph drawings}
Closely related are Lombardi graph drawings which are graph drawings with circular-arc edges and perfect angular resolution~\cite{DuncanEGKN10a}. 
Previous studies have demonstrated that a significant group of 4-regular planar graphs can be represented as plane Lombardi graph drawings~\cite{LD18}. However, there are certain restrictions. Notably, if a planar graph contains a loop, it cannot be depicted as a Lombardi drawing. In our approach, every 4-regular planar graph can be transformed in a knot (links).

\subsubsection{Celtic knots by tiling and algorithmic design methods}  
Celtic knots can be created using tiling and algorithmic design methods. George Bain introduced a formal method for creating Celtic knot patterns~\cite{Bain73}, which subsequently has been simplified to a three-grid system by Iain Bain~\cite{Bain86,Glassner99}. Klempien-Hinrichs and von Totth study the  generation of Celtic knots using collage grammars~\cite{Klempien-HinrichsT10}. And Even-Zohar et al.~investigate sets of planar curves which yield diagrams for all knots~\cite{osti}. None of those methods use graphs or are graph drawing approaches.

\subsubsection{Drawing graphs with B\'ezier curve edges} 
A number of network visualisation systems use B\'ezier curves as edges. These include \texttt{yWorks}~\cite{yWorks}, \texttt{GraphViz}~\cite{GraphViz}, \texttt{Vanted}~\cite{JunkerKS06b},  \texttt{Vizaj}~\cite{Vizaj}, and the framework proposed in~\cite{GoodrichW98}.
In many cases, such systems allow the user to route the curves by adjusting control points, but few provide automatic computation of the curves.
However, there are some exceptions.  For example, in the \texttt{GraphViz} system, B\'ezier curve edges are routed within polygons to avoid edge crossings~\cite{AbelloGansner}.
Force-directed methods are also popular for computing control points of B\'ezier curve edges~\cite{BrandesWagner,Finkel,Fink}.
Brandes et al.~present a similar method to the ``cross'' method in Sect.~\ref{se:BezierDrawing}, applied to transport networks~\cite{Brandes2000Impro-5663}. However, only~\cite{Fink} considers smoothness in more than one edge. None of those systems or approaches consider Celtic knots.

\section{Overview of the \texttt{CelticGraph} process}
\label{se:overview}

In this Section we outline \texttt{CelticGraph}, our framework for creating aesthetically pleasing pictures of 4-regular planar graphs as knots. 
The \texttt{CelticGraph} procedure is shown in Fig.~\ref{fi:processV2}; it has 5 steps:
    \begin{enumerate}
        \item [(a)] Create a topological embedding $G'$ of the input 4-regular planar graph $G$. 
        \item [(b)] Create a planar straight-line drawing $D$ of the plane graph $G'$.
        \item [(c)] Create a special circuit partition of $G'$, called a ``threaded circuit partition''.
        \item [(d)] Using the straight-line drawing $D$ and the threaded circuit partition $C$, create a drawing $D'$ of $G$ with cubic B\'ezier curves as edges. 
        \item [(e)] Render the drawing $D'$ as a knot, on the screen or with a 3D printer.
    \end{enumerate}

\begin{figure}[t]
  \centering
  \vspace*{-1ex}
  \includegraphics[width=\textwidth]{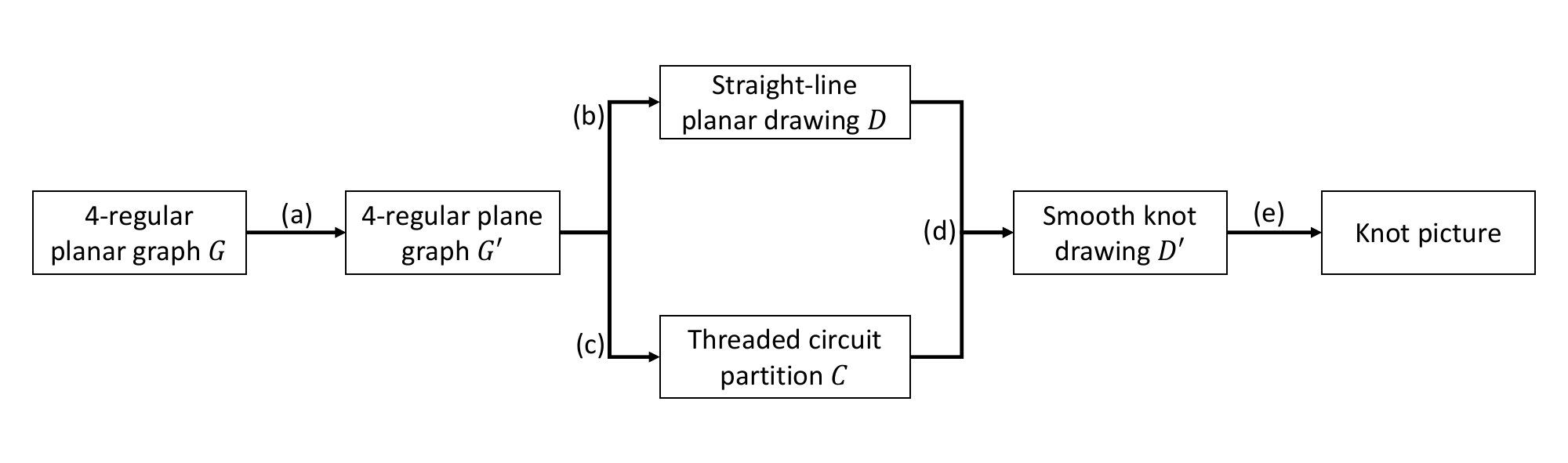}
       \vspace*{-5ex}
  \caption{\emph{The \texttt{CelticGraph} process}}
  \label{fi:processV2}
  \vspace*{-2ex}
\end{figure}

The first two steps can be done using standard Graph Drawing methods~\cite{DiBattistaETT99}. Steps (c) and (d) are described in the following Sections, step (e) can be done using standard rendering methods.

\section{Step (c): Finding the threaded circuit partition}\label{sec:c}

Here we define \emph{threaded circuit partition}, a special kind of circuit partition of a plane graph, and show how to find it in linear time.

A \emph{circuit} in a graph $G$ is a list of distinct edges $( e_0, e_1, \ldots, e_{k-1} )$ such that $e_i$ and $e_{i+1}$ share a vertex $i=0,1,\ldots, k-1$ (here, and in the remainder of this paper, indices in a circuit of length $k$ are taken modulo $k$.) We can write the circuit as a list of vertices $(u_0, u_1, \ldots, u_{k-1})$ where $e_i = (u_i, u_{i+1})$. Note that a vertex can appear more than once in a circuit, but an edge cannot.
A set $C = \{ c_0, c_1, \ldots, c_{h-1} \}$ of circuits in a graph $G$ such that every edge of $G$ is in exactly one $c_j$ is a \emph{circuit partition} for $G$. Given a circuit partition, we can regard $G$ as a directed graph by directing each edge so that each $c_i$ is a directed circuit.

A path $(\alpha,\beta,\gamma)$ of length two (that is, two edges $(\alpha,\beta)$ and $(\beta,\gamma)$) in a 4-regular plane graph $G$ is a \emph{thread} if edges $(\alpha,\beta)$ and $(\beta,\gamma)$ are not contiguous in the cyclic order of edges around $\beta$. This means that
there is an edge between $(\alpha,\beta)$ and $(\beta,\gamma)$ in both counterclockwise and clockwise directions in the circular order of edges around $\beta$.
We say that $\beta$ is the \emph{midpoint} of the thread $(\alpha,\beta,\gamma)$. Note that each vertex in $G$ is the midpoint of two threads; see Fig.~\ref{fi:threads}(a).
For every edge $(\alpha,\beta)$ in $G$, there is a unique thread $(\alpha,\beta,\gamma)$; we say that the edge $(\beta,\gamma)$ is the \emph{next edge after $(\alpha,\beta)$}.
For each vertex $u_j$ on a circuit $c = (u_0, u_1, \ldots , u_{k-1})$ with $k>1$ there is a path $p_j = (u_{j-1} , u_j, u_{j+1})$ of length two such that $u_j$ is the midpoint of $p_j$.  In fact we can consider that the circuit $c$ consists of $k$ paths of length two. We say that the circuit $c$ is \emph{threaded} if for each $j$, the path $p_j = (u_{j-1} , u_j, u_{j+1})$ is a thread. Note that in such a circuit, the edge $(u_j, u_{j+1})$ is the (unique) next edge after $(u_{j-1}, u_j)$ for each $j$.
A circuit partition $C = \{ c_0, c_1, \ldots, c_{h-1} \}$ is \emph{threaded} if each circuit $c_j$ is threaded. In the case that $h=1$, a threaded circuit partition defines a \emph{threaded Euler circuit}; see Fig.~\ref{fi:threads}(b).

\begin{figure}[t]
    \centering
    \vspace*{-1ex}
    \includegraphics[width=0.8\textwidth]{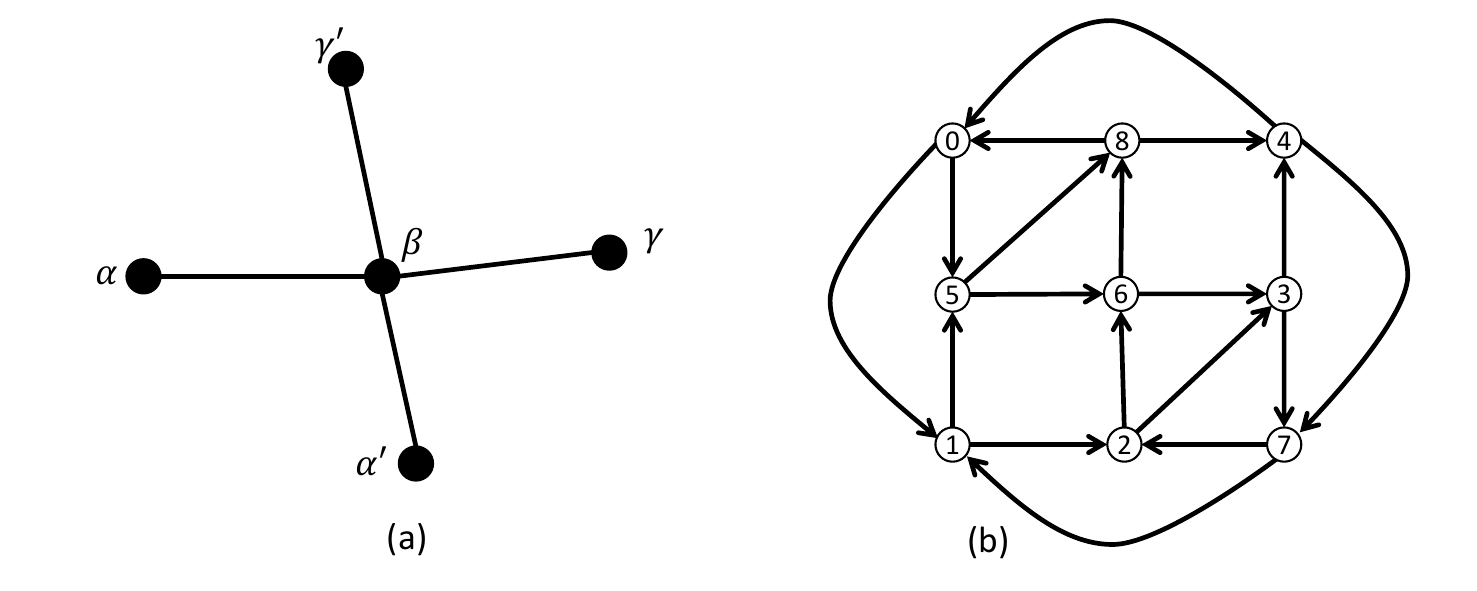}
    \caption{\emph{(a) Two threads, each with midpoint $\beta$. 
    (b) Plane 4-regular graph with a threaded Euler circuit $(0,1,2,3,4,0,5,6,3,7,1,5,8,4,7,2,6,8)$.}}
    \label{fi:threads}
    \vspace*{-2ex}
\end{figure}
    

An assignment $\upsilon (p) \in \{ -1 , +1 \}$ of an integer $-1$ or $+1$ to each thread $p$ of a 4-regular plane graph $G$ is an \emph{under-over assignment}.
Note that for each vertex $\beta$ of $G$, there are two threads $p_\beta$ and $p'_\beta$ with midpoint $\beta$. We say that an under-over assignment $\upsilon$ is \emph{consistent} if $\upsilon(p_\beta) = - \upsilon(p'_\beta)$ for each vertex $\beta$.

An under-over assignment $\upsilon$ is \emph{alternating} on the circuit $(p_0, p_1, \ldots , p_{k-1})$ if $\upsilon(p_i) = - \upsilon(p_{i+1})$ for each $i$.
An under-over assignment for a graph with a threaded circuit partition $C$ is \emph{alternating} if it is alternating on each circuit in $C$.

Intuitively, a consistent under-over assignment designates which thread passes under or over which thread, and an alternating under-over assignment corresponds to an alternating knot or link~\cite{Murasugi93}.



The following theorem gives the essential properties of threaded circuit partitions that are essential for \texttt{CelticGraph}.





\begin{theorem}
\label{th:threaded}
  Every 4-regular plane graph has a unique threaded circuit partition, and this threaded circuit partition has a consistent alternating under-over assignment. Further, this threaded circuit partition can be found in linear time.
\end{theorem}
\begin{proof}
The existence and uniqueness of the threaded circuit partition follows  from the fact that every edge has a unique next edge.
A simple linear-time algorithm to find the threaded circuit partition is to repeatedly choose an edge $e$ that is not currently in a circuit, then repeatedly choose the next edge after $e$ until we return to $e$.
We can direct every edge of a 4-regular planar graph $G$ so that each circuit in a given threaded circuit partition $C$ is a directed circuit.
This means that we can sensibly define the ``left'' and ``right'' faces of an edge. Since a 4-regular plane graph is bridgeless~\cite{DBLP:books/ws/NishizekiR04}, no face is both ``left'' and ``right''.

\begin{figure}[t]
    \centering
    \vspace*{-1ex}
    \includegraphics[width=0.8\textwidth]{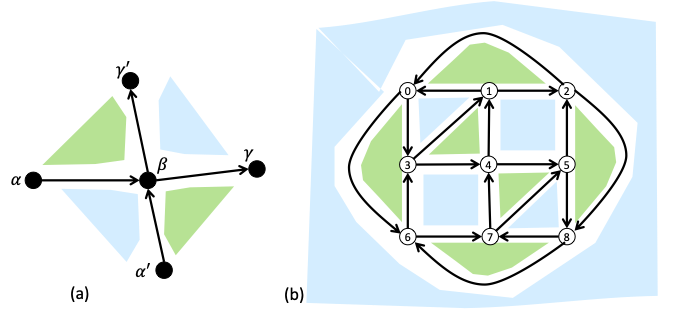}
         \vspace*{-2ex}
    \caption{\emph{(a) Two threads: $(\alpha,\beta,\gamma)$ has under-over assignment $+1$ (since the face to the left of $(\alpha,\beta)$ is green), and $(\alpha',\beta,\gamma')$ has under-over assignment $-1$ (since the face to the left of $(\alpha',\beta)$ is blue). (b) The faces of the graph are coloured according to its bipartition; note that each vertex has two incoming edges: one has a blue face to the left, the other has a green face to the left, and that the faces on the left of the threaded Euler circuit alternate in colour.}}
    \label{fi:threadsColoured}
    \vspace*{-2ex}
\end{figure}

Since the planar dual graph of a 4-regular planar graph is bipartite~\cite{DBLP:books/ws/NishizekiR04},
the faces can be coloured \emph{green} and \emph{blue}, such that no two faces of the same colour share an edge, see Fig.~\ref{fi:threadsColoured}. 
An immediate consequence is that the sequence of left faces to (directed) edges in a threaded circuit \emph{alternate} in colour. 
Now consider a thread $(\alpha,\beta,\gamma)$
in a (directed) threaded circuit in a threaded circuit partition. If the face to the left of $(\alpha,\beta)$ is green, then assign $+1$ to the path $(\alpha,\beta,\gamma)$; otherwise assign $-1$ to $(\alpha,\beta,\gamma)$.
Note that the face to the left of $(\beta,\gamma)$ is the opposite colour of the face to the left of $(\alpha,\beta)$, and so the under-over assignment is alternating. Further it is consistent, since at each vertex there is precisely one incoming arc with a green face on the left, and precisely one incoming arc with a blue face on the left.
\end{proof}

\subsubsection{Threaded Euler circuits}
Celtic knots are sometimes called ``endless knots'', and can be used to symbolise eternity. For this reason, a threaded \emph{Euler} circuit is desirable; such a circuit gives a drawing of the graph as a knot rather than a link. Using the algorithms in the proof of Theorem~\ref{th:threaded}, one can test whether a given plane graph has a threaded Euler circuit in linear time.
Note that different topological embeddings of a given planar graph may have different threaded circuit partitions; see Fig.~\ref{fi:differentLengthsV2}. It is clear that, in some cases, we can increase the length of a threaded circuit by changing the embedding. It is tempting to try to find a method to adjust the embedding to get a threaded Euler circuit. However, it can be shown that changing the embedding cannot change the \emph{number} of threaded circuits in a threaded circuit partition; see Appendix \ref{appendix:cardinality}.

\begin{figure}[t]
    \centering
    \vspace*{-1ex}
    \includegraphics[width=0.55\textwidth]{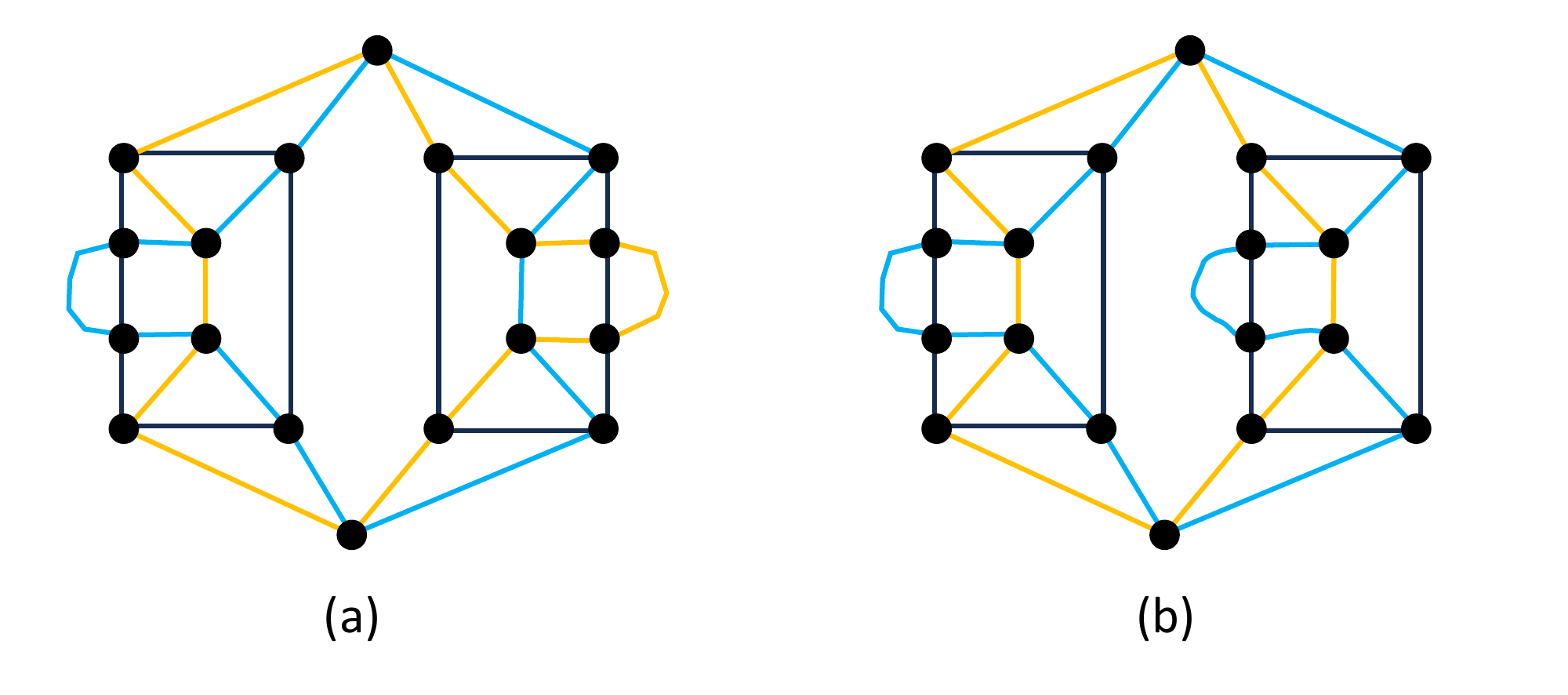}
         \vspace*{-2ex}
    \caption{\emph{Two topological embeddings of a planar graph. In (a), the plane graph has a threaded circuit partition of 4 circuits, with two circuits of size 6 (in black) and two circuits of size 12 (in blue and orange). In (b), the plane graph still has 4 threaded circuits: the two of size 6 are unchanged, but the lengths of the blue and orange circuits are 14 and 10 respectively.}}
    \label{fi:differentLengthsV2}
    \vspace*{-2ex}
\end{figure}

\section{Step (d): Smooth knot drawing with B\'ezier curves}
\label{se:BezierDrawing}

Step (d) takes a straight-line drawing $D$ of the input graph $G$, and replaces the straight-line edges by cubic B\'ezier curves in a way that ensures that each circuit in the threaded circuit partition found in step (c) is smooth. 

A central concept for the smooth drawing method is a ``cross'' $\chi_u$ at each vertex $u$. For each $u$, $\chi_u$ consists of 4 line segments called ``arms''. The four arms are all at right angles to each other, leading to a perfect angular resolution. Each arm of $\chi_u$ has an endpoint at $u$.
\begin{figure}[t]
    \centering
    \vspace*{-1ex}
    \includegraphics[width=0.7\textwidth]{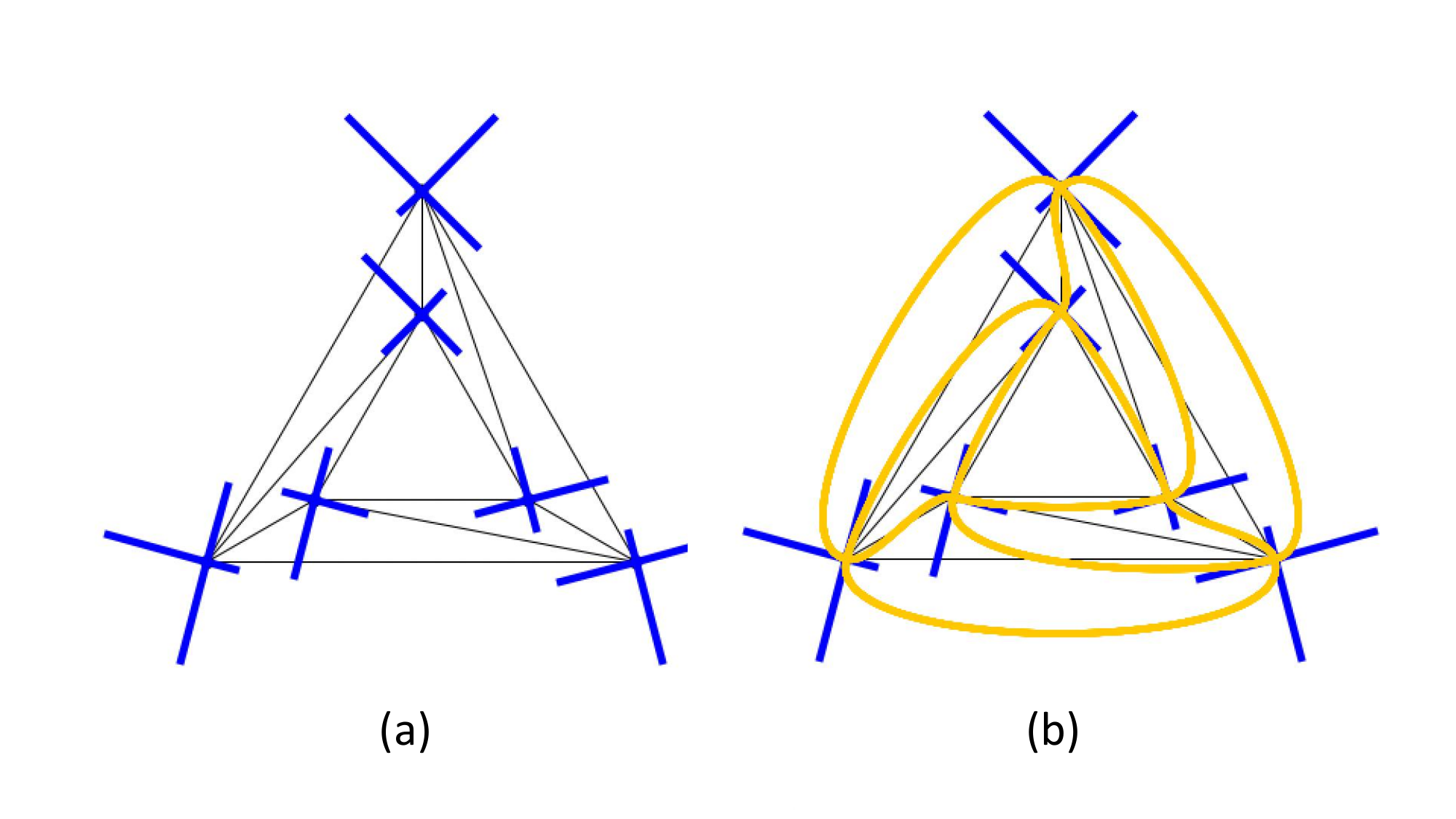}
    \caption{\emph{(a) The 3-prism with a cross at each vertex. (b) Edges drawn as B\'ezier curves, using the arms of the crosses as control tangents.}}
    \label{fi:crosses2}
    \vspace*{-2ex}
\end{figure}
This is illustrated in Fig.~\ref{fi:crosses2}(a).
Each edge $(u,v)$ then is drawn as a cubic B\'ezier curve with endpoints $u$ and $v$, and the control tangents of the curve are arms of the crosses $\chi_u$ and $\chi_v$ (illustrated in Fig.~\ref{fi:crosses2}(b)).

For this approach, we need to choose three parameters for each cross $\chi_u$:
\begin{enumerate}
    \item The mapping between the four arms of $\chi_u$ and the four edges incident to $u$.
    \item The angle of orientation of the cross.
    \item The length of each arm of the cross.
\end{enumerate}
These parameters are discussed in the next subsections. The methods described in Sect.~\ref{se:edgeArmMapping} and \ref{se:orientation} are analogous to the methods in~\cite{Brandes2000Impro-5663}; Sect.~\ref{se:armLengths} is not.

\subsection{The edge-arm mapping}
\label{se:edgeArmMapping}

Suppose that $u$ is a vertex in the straight-line drawing $D$ of the input plane graph.
We want to choose the mapping between the arms of the cross $\chi_u$ and the edges incident with $u$ so that the arms are approximately in line with the edges. 

Now suppose that the edges incident with $u$ are $e_0, e_1, e_2, e_3$ in counterclockwise order around $u$.
For each $i=0,1,2,3$ we choose an arm $\alpha_{i}$ of the cross $\chi_u$ corresponding to $e_i$ so that the  counterclockwise order of arms around $u$ is the same as the order of edges around $u$; that is, the counterclockwise order of arms is
$\alpha_{0}, \alpha_{1}, \alpha_{2}, \alpha_{3}$. Note that this method separates multi-edges.

\subsection{The orientation of the cross}
\label{se:orientation}

To improve the alignment of the arms of the crosses with the edges, we rotate each cross.
Suppose that the counterclockwise angle that edge $e_i$ makes with the horizontal direction is $\phi_i$. We want to rotate the cross by an angle $\theta$ to align with the edges, as best as possible. This is illustrated in Fig.~\ref{fi:anglesV3}.

\begin{figure}[t]
    \centering
    \vspace*{-1ex}
    \includegraphics[width=0.7\textwidth]{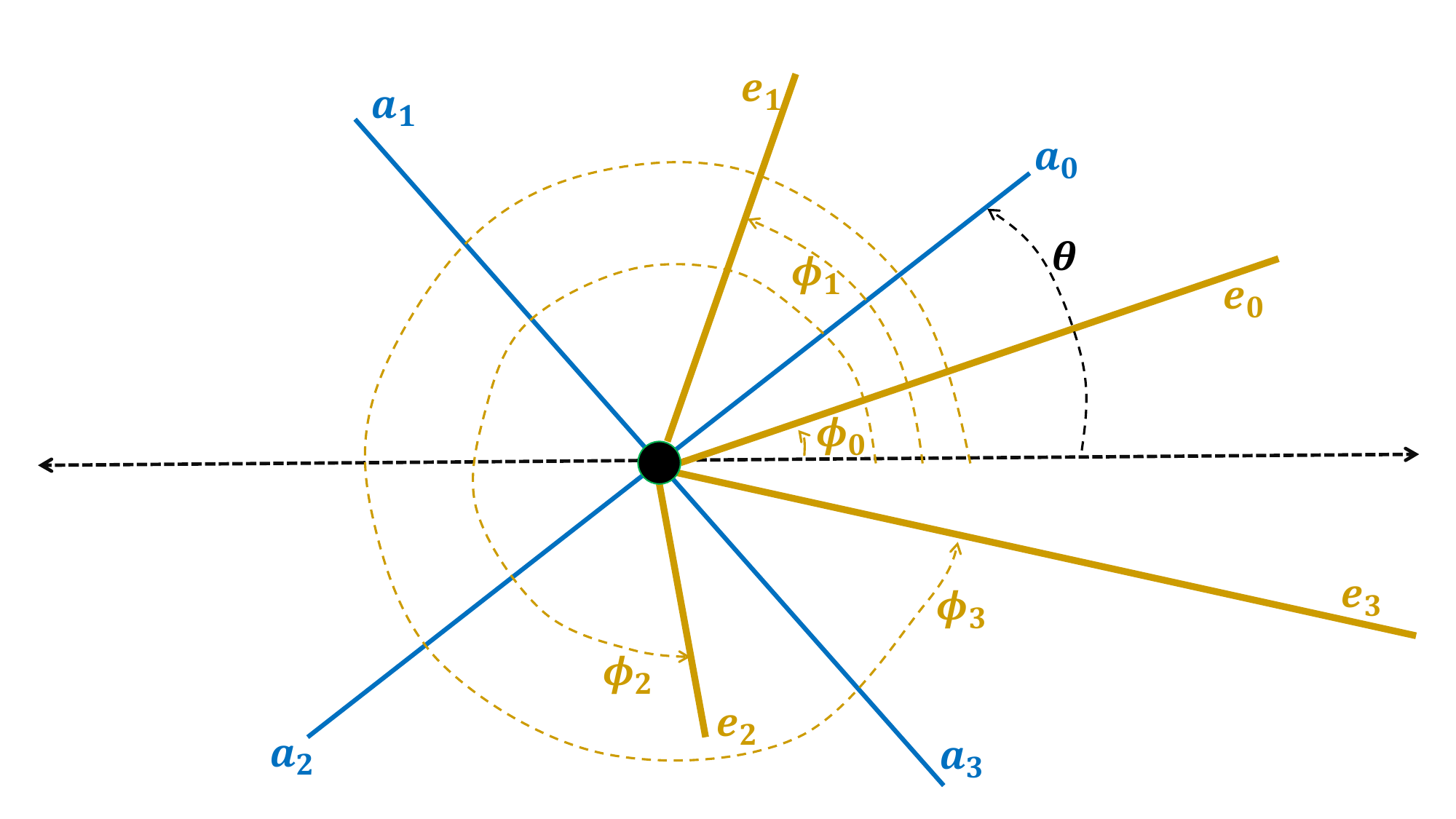}
    \caption{\emph{A cross rotated by an angle of $\theta$. Here the cross is in blue, the edges of the graph are in orange.}}
    \label{fi:anglesV3}
    \vspace*{-2ex}
\end{figure}

Consider the sum of squares error in rotating by $\theta$; this is:
\begin{equation}
    f(\theta) = \sum_{i=0}^{i=3} \left( \theta + \frac{i\pi}{2} - \phi_i \right)^2 .
\end{equation}

To minimise $f(\theta)$, we solve $f'(\theta) = 0$ and choose the optimum value:
\begin{equation}
    \theta^* = \frac{1}{4} \left( \sum_{i=0}^{i=3} \phi_i \right) - \frac{3 \pi}{4}.
\end{equation}

In Fig.~\ref{fi:orientation}, we show a graph with crosses oriented by this method.

\begin{figure}[t]
    \centering
    \vspace*{-1ex}
    \includegraphics[width=\textwidth]{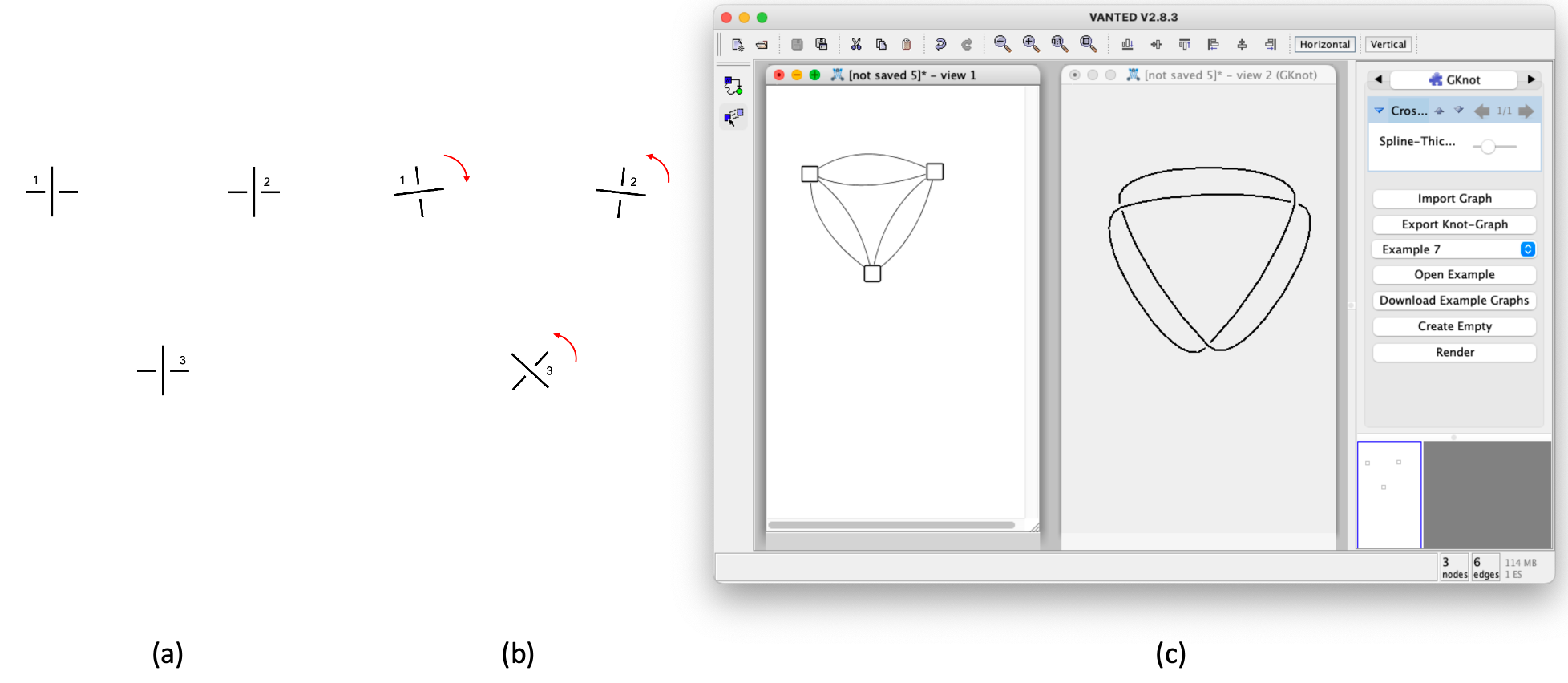}
    \caption{\emph{A graph with crosses oriented to align with edges as much as possible before (a) and after (b) applying the algorithm, and shown in \texttt{Vanted} (c).}}
    \label{fi:orientation}
    \vspace*{-2ex}
\end{figure}

\subsection{Arm length}
\label{se:armLengths}

\begin{figure}[t]
    \centering
    \includegraphics[width=0.95\textwidth]{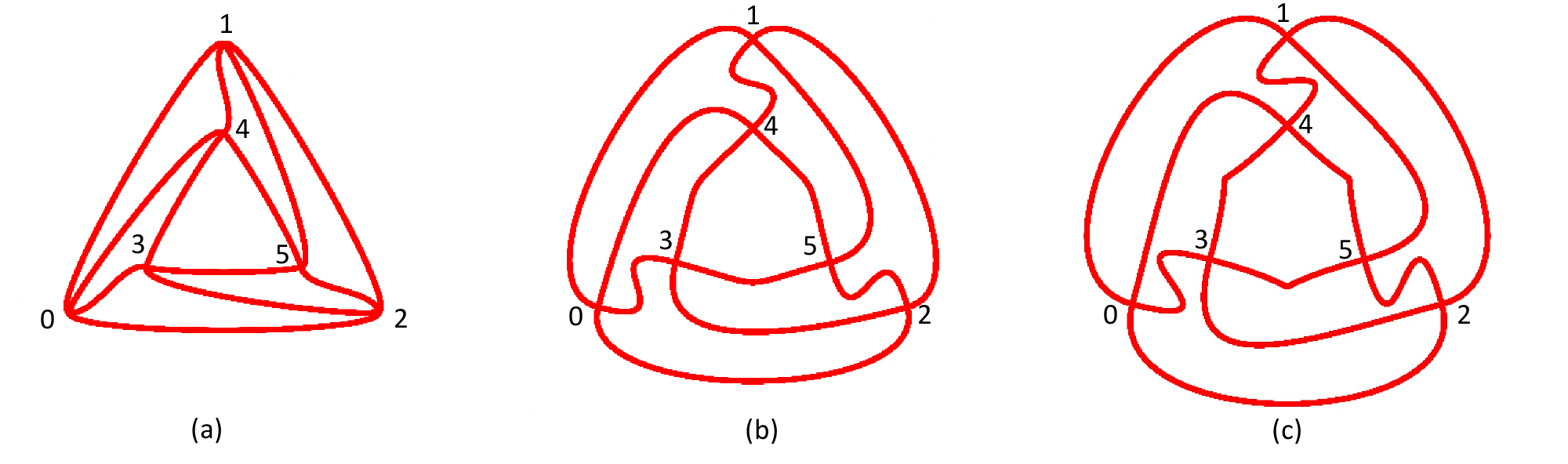}
         \vspace*{-2ex}
    \caption{\emph{Three drawings of the 3-prism, differing in edge curvature.}}
    \label{fi:curvature1}
    \vspace*{-2ex}
\end{figure}

Note that the ``apparent smoothness'' of an edge depends on its curvature. We illustrate this with Fig.~\ref{fi:curvature1}, which shows three B\'ezier curve drawings of the 3-prism. This graph has 3 threaded circuits, and we want to draw it so that each one of these threaded circuits appears as a smooth curve with limited curvature.
In Fig.~\ref{fi:curvature1}(a), the arms of the crosses are all very short, resulting in a B\'ezier curve drawing which is very close to a straight-line drawing. Each edge has low curvature in the middle and high curvature around the endpoints. The high curvature near their endpoints results in a lack of apparent smoothness where two B\'ezier curves join (at the vertices); it is difficult to discern the three threaded circuits. 
The arms of the crosses are longer in Fig.~\ref{fi:curvature1}(b), resulting in better curvature at the endpoints. However, here each of the edges $(0,3), (1,4)$, and $(2,5)$ have two points of large curvature; this is undesirable.
In Fig.~\ref{fi:curvature1}(c), the arms of the crosses are longer still. Each of the edges $(0,3), (1,4)$, and $(2,5)$ again have two points of large curvature, but the edges $(3,4),(4,5)$, and $(5,3)$ are worse: each has a ``kink'' (a point of very high curvature, despite being $C^1$-smooth).

Next we describe three approaches to choosing the lengths of the arms of the crosses, aiming to give sufficiently small curvature.

\subsubsection{Uniform arm lengths}
The curvature of the edge varies with lengths of the arms, and we want to ensure that the maximum curvature in each edge is not too large. The simplest approach is to use \emph{uniform arm lengths}, that is, judiciously choose a global value $\lambda$ and set the length of every arm length to $\lambda$. The drawings of the 3-prism in Fig.~\ref{fi:curvature1} have uniform arm lengths: $\lambda$ in Fig.~\ref{fi:curvature1}(a) is quite small, in Fig.~\ref{fi:curvature1}(c) it is relatively large, and (b) is in between. In fact, the problem with the uniform arm length approach is typified in Fig.~\ref{fi:curvature1}: if $\lambda$ is small, the curvature is high near the endpoints for all edges, and increasing $\lambda$ increases the curvature away from the endpoints, especially in the shorter edges. There is no uniform value of $\lambda$ that gives good curvature in both short and long edges.

\subsubsection{Uniformly proportional arm lengths}
 An approach that aims to overcome the problems of uniform arm length is to use \emph{uniformly proportional arm lengths}: we judiciously choose a global value $\alpha$, and then set the lengths of the two arms for edge $(u,v)$ to $\alpha d(u,v)$, where $d(u,v)$ is the Euclidean distance between $u$ and $v$. 
Fig.~\ref{fi:curvature2} shows typical results for the uniformly proportional approach. For $\alpha = 0.2$ the drawing is similar to Fig.~\ref{fi:curvature2}(a), and has similar problems. But for values of $\alpha$ near $0.5$ (Fig.~\ref{fi:curvature2}(b) and (c)), we have acceptable results; in particular, the shorter edges have acceptable curvature.

\begin{figure}[t]
    \centering
    \vspace*{-1ex}
    \includegraphics[width=0.95\textwidth]{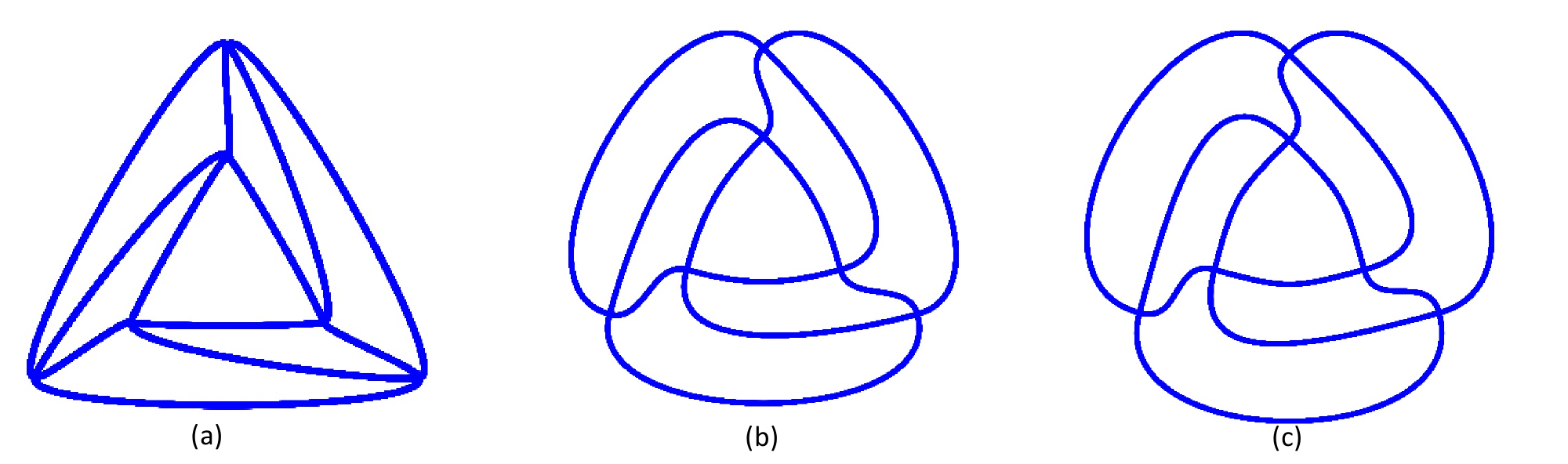}
         \vspace*{-1ex}
    \caption{\emph{The uniformly proportional approach: (a) $\alpha = 0.2$; (b) $\alpha = 0.4$; (c) $\alpha = 0.6$.}}
    \label{fi:curvature2}
    \vspace*{-3ex}
\end{figure}

\subsubsection{Optimal arm lengths}

\begin{figure}[t]
    \centering
    \vspace*{-1ex}
    \includegraphics[width=0.7\textwidth]{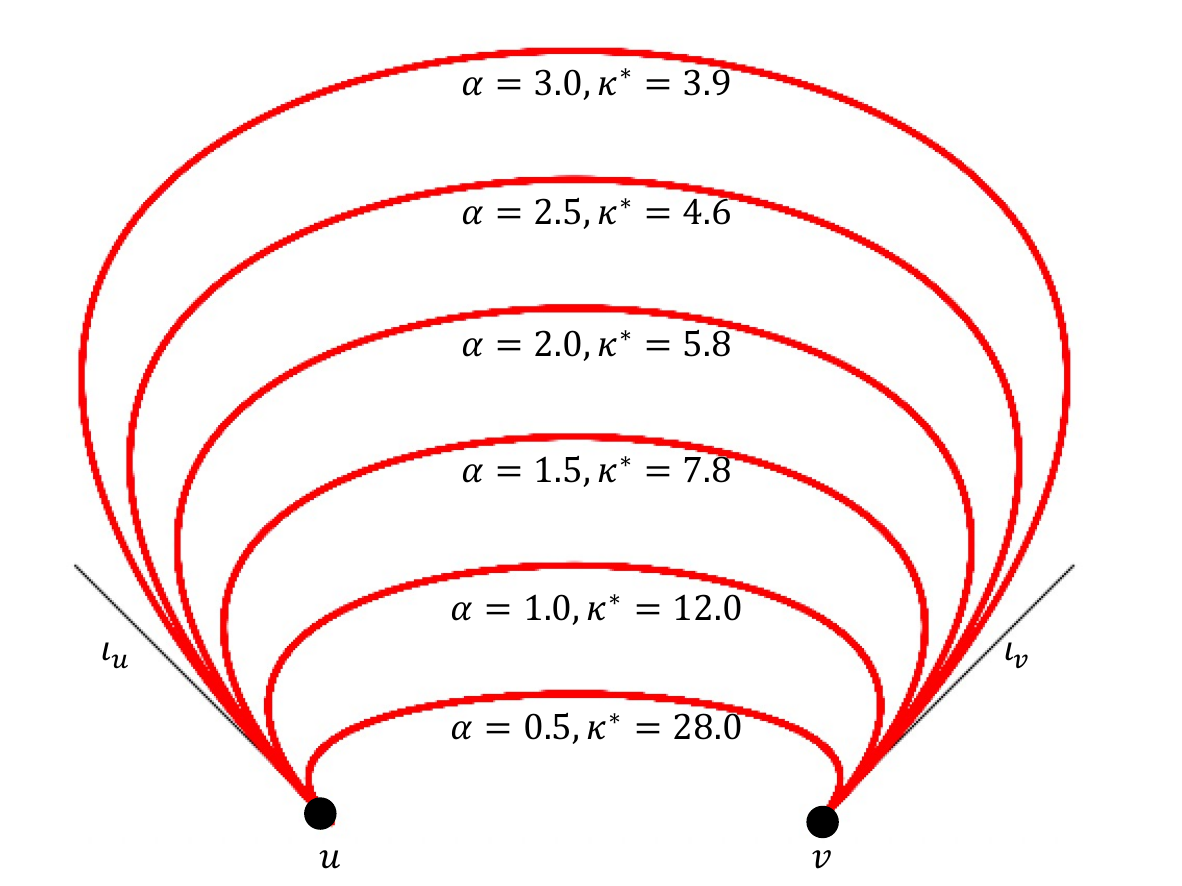}
    \caption{\emph{``Ballooning'' curves: as $\lambda_u$ and $\lambda_v$ increase, curvature falls but the curve becomes very long.}}
    \label{fi:curvature3}
    \vspace*{-2ex}
\end{figure}

A third approach is to choose the arm lengths at each end of an edge $(u,v)$ to minimise maximum curvature, as follows.
Suppose  $\kappa(t, \lambda_u , \lambda_v )$ is the curvature of the edge $(u,v)$ at point $t$ on the curve, when the arm lengths are $\lambda_u$ and $\lambda_v$ at $u$ and $v$ respectively. From Equation~(\ref{eq:curvature}), we note that

\begin{figure}[t]
    \centering
    \vspace*{-3ex}
    \includegraphics[width=1\textwidth]{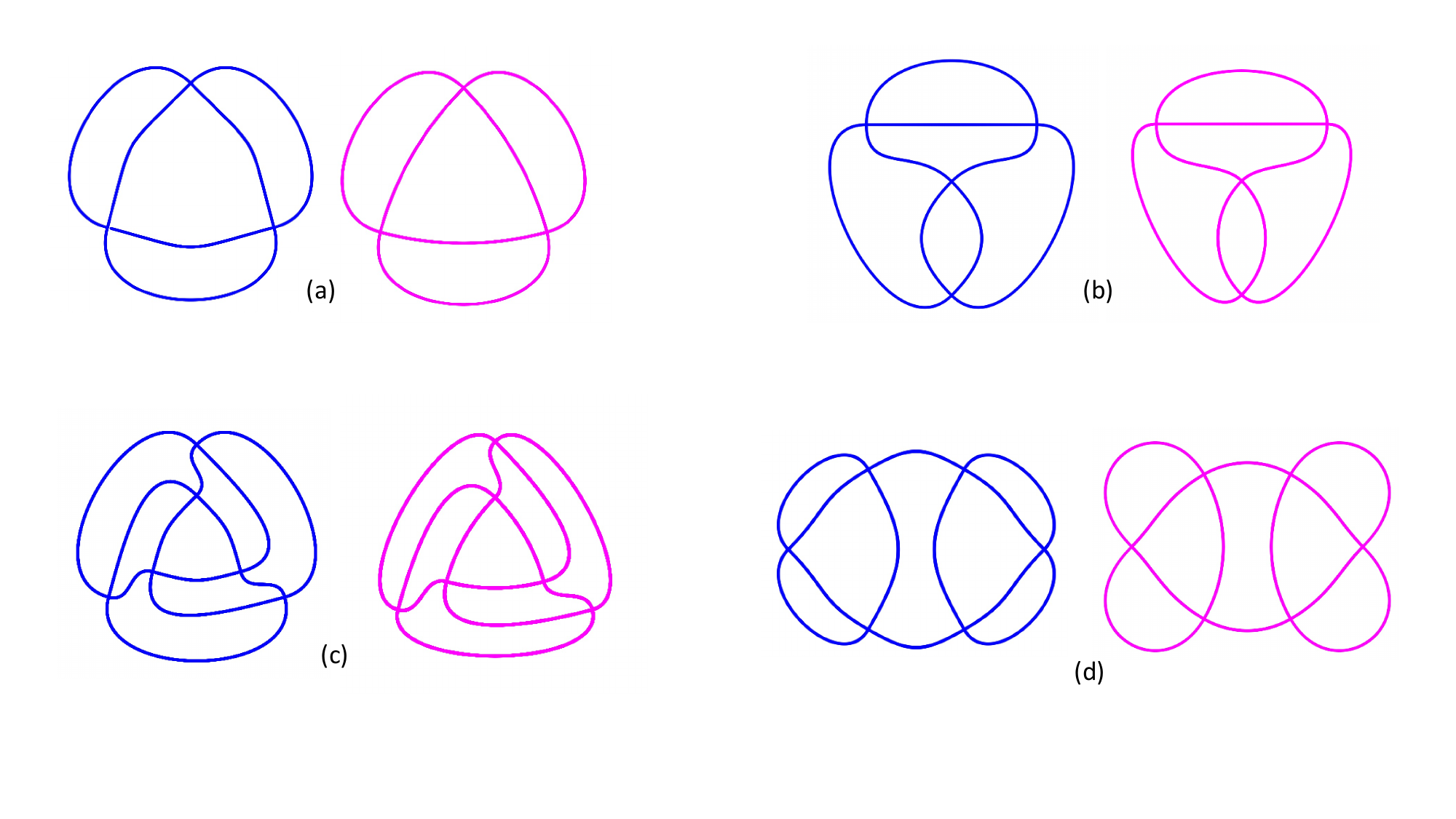}
         \vspace*{-10ex}
    \caption{\emph{Comparison of proportional arm length (blue) and optimal arm length (magenta): (a) Trefoil; (b) $K_4$ knot; (c) 3-prism; (d) Love knot.}}
    \label{fi:curvature4}
    \vspace*{-2ex}
\end{figure}

\begin{equation}
\label{eq:curvatureDerivative}
    \frac{\partial}{\partial t} \kappa(t, \lambda_u , \lambda_v )
    = \left| \frac{\dddot{x}\dot{y} - \dddot{y}\dot{x}}{(\dot{x}^2 + \dot{y}^2)^{1.5}}
    - \frac{\ddot{x} \dot{y} - \ddot{y} \dot{x}}{3(\ddot{x} + \ddot{y})^{2.5}} \right|
\end{equation}
as long as $(\dot{x}^2 + \dot{y}^2) \neq 0$ and $\ddot{x}\dot{y} \neq \ddot{y}\dot{x}$.
Since both $x$ and $y$ are cubic functions of $t$, equation~(\ref{eq:curvatureDerivative}) is not as complex as it seems, and it is straightforward (but tedious, because of the edge cases) to maximise  $\kappa(t, \lambda_u , \lambda_v )$ over $t$; that is, to find the maximum curvature $\kappa^*(\lambda_u, \lambda_v)$:
\begin{equation*}
    \label{eq:maxKappa}
    \kappa^* (\lambda_u, \lambda_v) = \max_{0 \leq t \leq 1} \kappa(t, \lambda_u , \lambda_v ).
\end{equation*}
Now we want to choose the arm lengths $\lambda_u$ and $\lambda_v$ to minimise $\kappa^* (\lambda_u, \lambda_v)$. Suppose that the unit vectors in the directions of the appropriate arms of $\chi_u$ and $\chi_v$ are $\iota_u$ and $\iota_v$ respectively. Note that we can express the internal control points $p_1$ and $p_2$ of the B\'ezier curve in terms of $\lambda_u$ and $\lambda_v$:
\begin{equation*}
    \label{eq:lambda}
    p_1 = (1 - \lambda_u) u + \lambda_u \iota_u,~~ p_2 = (1 - \lambda_v) v + \lambda_v \iota_v.
\end{equation*}
In this way, $\kappa^* (\lambda_u, \lambda_v)$ is linear in both $\lambda_u$ and $\lambda_v$ and finding a minimum point for  $\kappa^* (\lambda_u, \lambda_v)$ is straightforward.
However, in some cases, an edge with globally minimum maximum curvature may not be desirable. In Fig.~\ref{fi:curvature3}, for example, the curvature decreases as $\lambda_u$ and $\lambda_v$ increase; for large values of  $\lambda_u$ and $\lambda_v$ the curvature is quite low. The problem is that these large values make the curve very long (it ``balloons'' out), which might also cause unintended edge crossings.

For this reason, we choose an upper bounds $\epsilon_u$ and $\epsilon_v$ and take a minimum constrained by $\lambda_u \leq \epsilon_u$ and $\lambda_v \leq \epsilon_v$:
\begin{equation*}
    \label{eq:kappanMin}
    \kappa^*_{\min} = \min_{0 \leq \lambda_u \leq \epsilon_u, 0 \leq \lambda_v \leq \epsilon_v } \kappa^*(\lambda_u, \lambda_v).
\end{equation*}
We have found that $\epsilon_u = \epsilon_v = 0.75 d(u,v)$ gives good results, where $d(u,v)$ is the distance between $u$ and $v$. Values of $\lambda_u$ and $\lambda_v$ that achieve the (constrained) minimum $\kappa^*_{\min}$ are then used by the B\'ezier curves.
In practice, using such optimal arm lengths gives better results than using uniformly proportional arm lengths. In some cases the difference is not significant, but in others the optimal edges appear to be much smoother. See Fig.~\ref{fi:curvature4} for examples.

\begin{figure}[t]
    \centering
    \includegraphics[width=0.45\textwidth]{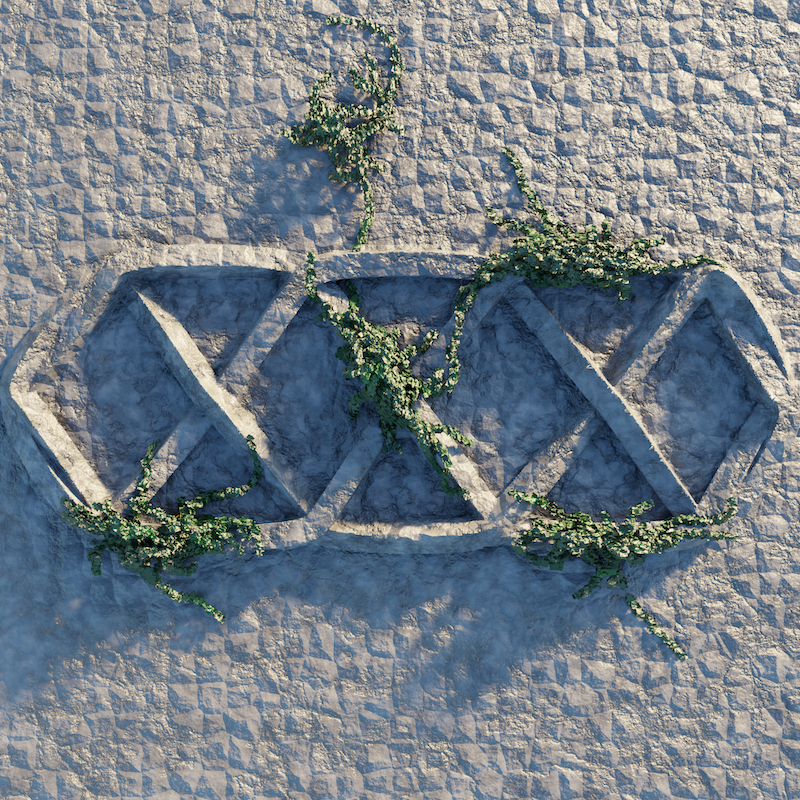}
    \includegraphics[width=0.45\textwidth]{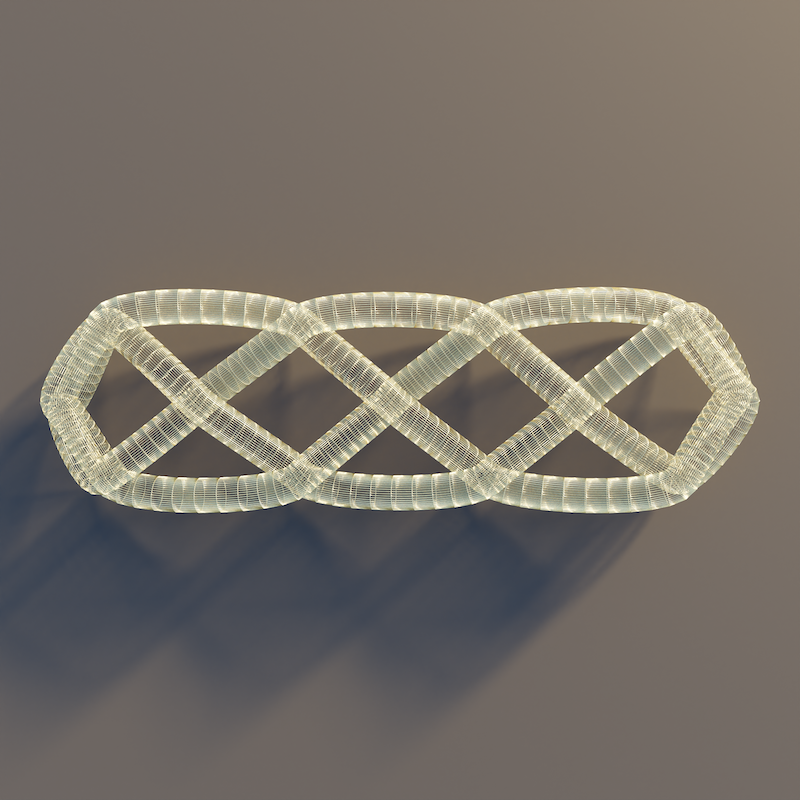}
    \includegraphics[width=0.45\textwidth]{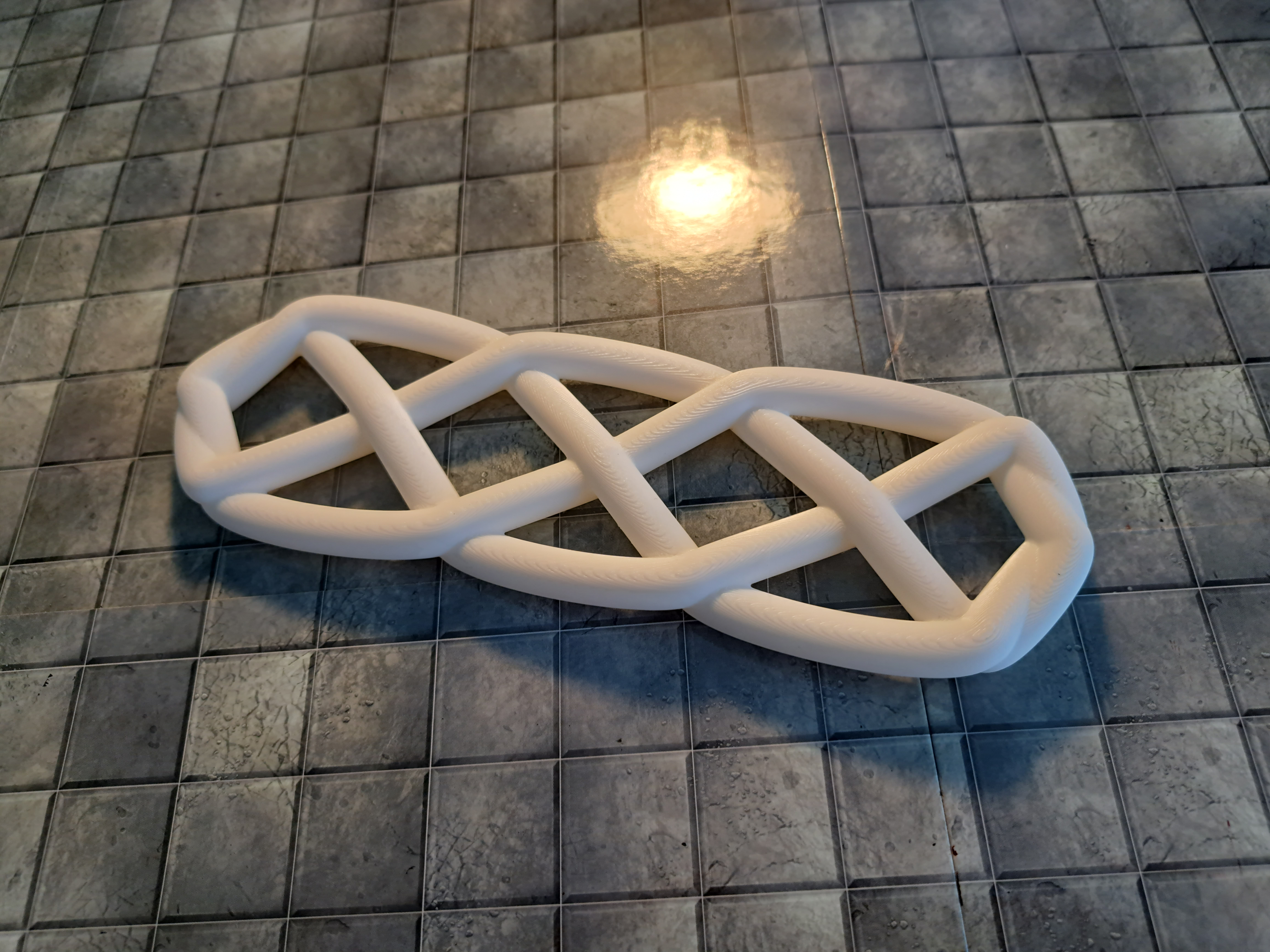}
    \includegraphics[width=0.45\textwidth]{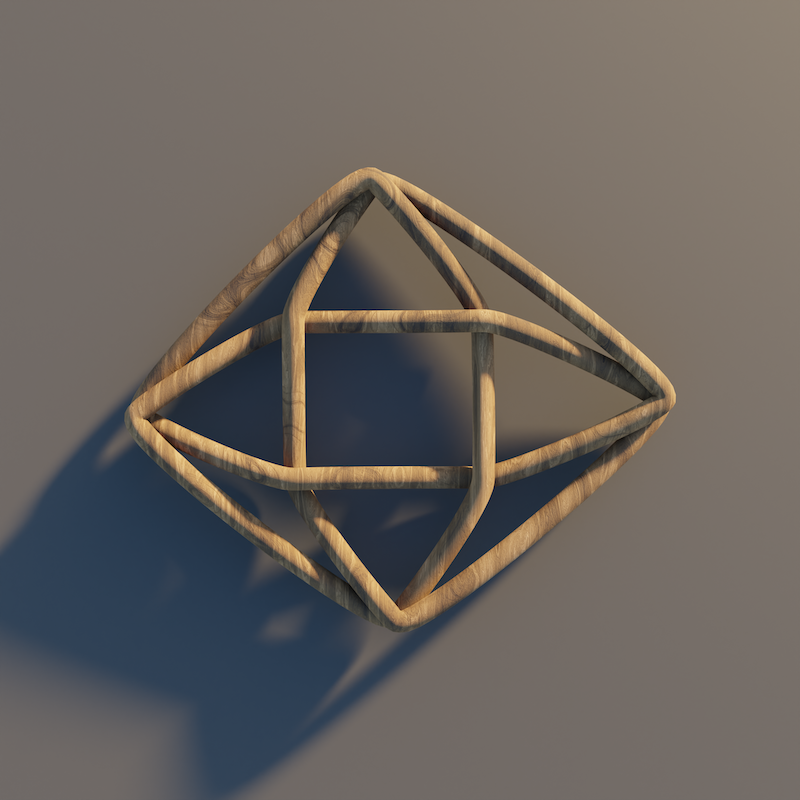}
    \includegraphics[width=0.45\textwidth]{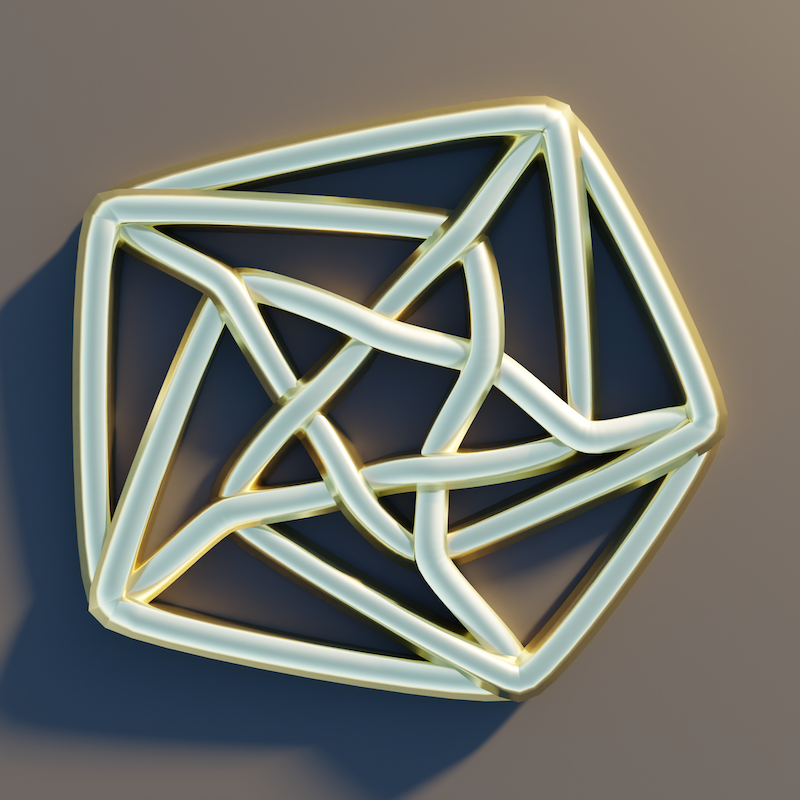}
    \includegraphics[width=0.45\textwidth]{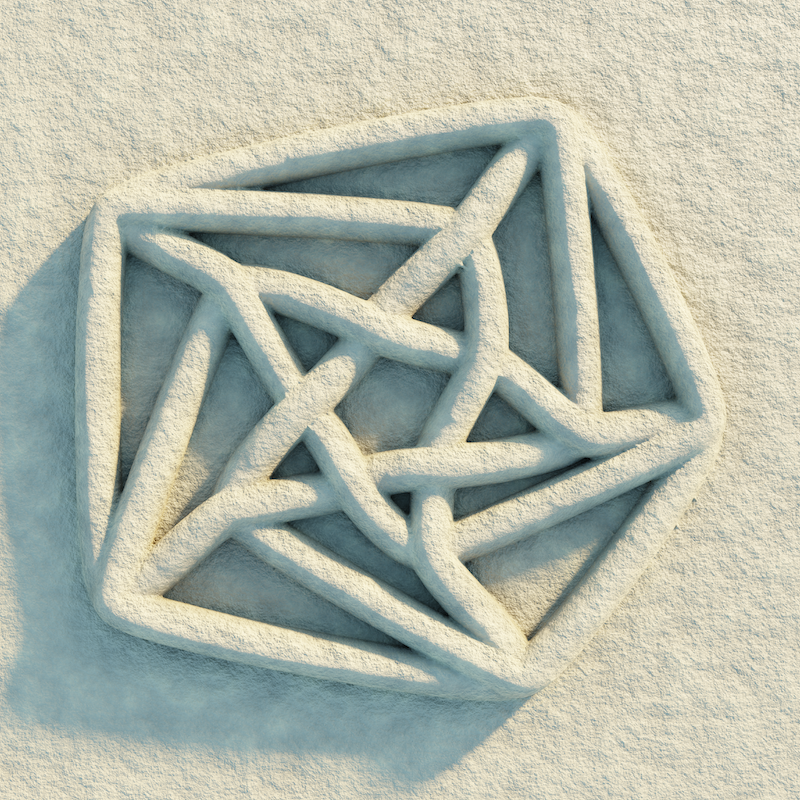}
    \caption{Examples of Celtic knot renderings in different media including a 3D printed version (mid left).}
    \label{fi:gallery}
\end{figure}

\section{\texttt{CelticGraph} implementation as a \texttt{Vanted} add-on and rendering}

\texttt{CelticGraph} has been implemented as an add-on of \texttt{Vanted}, a tool for interactive visualisation and analysis of networks. Figure~\ref{fi:teaserV2} shows an example workflow; the first step is implemented as \texttt{Vanted}~\cite{JunkerKS06b} add-on, the second is done by Blender~\cite{Blender}.

\texttt{Vanted} allows a user to load or create 4-regular graphs, either by importing from files (e.\,g.~a .gml file), by selecting from examples, or by creating a new graph by hand. The individual vertices of the graph are then mapped into the data structure of a cross, containing position and the rotation and control points of the to-be-generated B\'ezier curves. The graph is  translated into a knot (link) using the methods for optimal cross rotation and arm length computation  described in the previous sections. Vertex positions can be interactively changed, either by interacting with the underlying graph, or by interacting directly with the visualisation of the knot. Once a visualisation satisfies the expectation of the user, the B\'ezier curves can be exported for further use in Blender.

We implemented a Python script and a geometry node tree in Blender which allows importing the information into Blender and rendering the knot (links), either using a set of predefined media or interactively; the script can also run as batch process with selected parameters and media. Figure~\ref{fi:gallery} shows examples of Celtic knot renderings in different media such as in metal, in stone, with additional decoration and so on;  knots can be also printed in 3D. More examples can be found in the gallery of our web page \url{http://celticknots.online} which also provides the \texttt{Vanted} add-on, Blender file and a short manual. 





\subsection*{Acknowledgements}
Partly funded by the Deutsche Forschungsgemeinschaft (DFG, German Research Foundation) – Project-ID 251654672 – TRR 161.

\clearpage

\bibliographystyle{splncs04}
\bibliography{bib}

\FloatBarrier
\newpage

\appendix
\section{The cardinality of a threaded circuit partition is fixed}

\label{appendix:cardinality}

In this appendix we provide a proof for the claim made above in Section~\ref{sec:c};
namely that the \emph{number} of threaded circuits in any threaded circuit partition
of a 4-regular plane graph $G$ is fixed by the combinatoric structure of $G$, and cannot
be changed by taking a different plane embedding of $G$.

Since there is no choice of plane embedding for 3-connected graphs~\cite{whitney}
we will be concerned only with 1- and 2-connected graphs. We will not consider disconnected graphs;
in that case each connected component should be treated separately. Clearly a threaded circuit cannot
span multiple connected components, so the cardinality of a threaded circuit partition of a disconnected
graph is exactly the sum of the cardinalities of the threaded circuit partitions of its components.

Suppose that $G$ is a 4-regular planar graph, and $c$ is a cutpoint separating components $G_1$ and $G_2$ in $G$. Note that $c$ has degree 2 in both $G_1$ and $G_2$, because a 4-regular graph is bridgeless.
Now suppose that the two edges incident with $c$ in $G_1$ are $e_1$ and $e_2$. By the handshake lemma $e_1$ and $e_2$ belong to the same threaded circuit in any threaded circuit of an embedding of $G$. Thus, noting that $c$ is the midpoint of two threads, all four edges incident with $c$ belong to the same threaded circuit. Therefore changing the order of 2-connected components around a cutpoint does not change the number of threaded circuits; in fact we can deduce the following theorem.

\begin{theorem}
Suppose that $G$ is a 4-regular planar graph, and $c$ is a cutpoint separating components $G_1$ and $G_2$ in $G$; denote by $G'_i$ the 4-regular planar graph formed by adding a self-loop to $c$ in $G_i$, for $i=1,2$. Suppose that $\hat{G}$ is a planar embedding of $G$, and $\hat{G'_1}$ and $\hat{G'_2}$ are plane subgraphs of $G$ corresponding to $G'_1$ and $G'_2$ respectively. Suppose that $\hat{G},\hat{G_1}$ and $\hat{G_2}$ have threaded circuit partitions $C, C_1$, and $C_2$ respectively. Then $|C| = |C_1| + |C_2| - 1$.
\end{theorem}

In contrast, ``mirroring'' about a \emph{separation pair} can change the threaded circuits; as demonstrated in
Fig.~\ref{fi:differentLengthsV2}. However, we will now show that the cardinality of a threaded circuit partition is a \emph{graph property}, not an embedding property; that is, changing the embedding does not change the \emph{number} of threaded circuits in a threaded circuit partition.

\begin{theorem}
    Suppose that $G$ is a 2-connected 4-regular planar graph, and $C$ and $C'$ are threaded circuit partitions of two planar embeddings of $G$.
    Then $|C| = |C'|$.
\end{theorem}
\begin{proof}
    For a 2-connected planar graph, all planar embeddings can be found by ``pivoting'' (re-ordering and ``mirroring'') 3-connected components around separation pairs. We show that such a pivoting does not change the cardinality of a threaded circuit partition.
    
    Firstly note that pivoting does not change the internal structure of any component; it only changes
     the ``interface'', that is, the way in which it connects to the rest of the graph. Thus the only threaded circuits that are possibly changed are those that pass through one or both vertices in the separation pair.

    Now suppose that vertices $c$ and $d$ form a separation pair that separates $G_1$ from $G_2$. The handshake lemma implies that either both $c$ and $d$ have degree 2 in $G_1$, or both have odd degree in $G_1$. One can deduce that, since the graph is 4-regular and planar, there are only four possible ``interfaces'', enumerated in Fig.~\ref{img:components}.

    To complete the proof, note that pivoting each type of subgraph does not change the cardinality of the threaded circuit partition.

\begin{figure}
  \centering
  \includegraphics[width=0.6\textwidth]{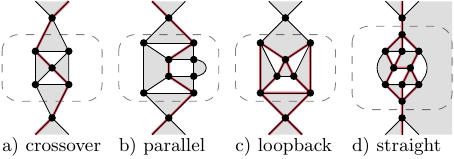}
  \caption{\emph{Examples of the four types of subgraphs separated by a separation pair. The type of a subgraph
   is determined by a perfect matching on the edges in its cut; there are
  three ways to match four edges, and one way to match two edges.}}
  \label{img:components}
\end{figure}
\end{proof}

We finally note that this observation provides an extension to the simple linear-time algorithm in Section~\ref{sec:c} that tests whether a graph has a threaded Euler circuit (in any embedding). Using a SPQR tree~\cite{SPQR}, one can decompose the graph into pivotable components and solve the following optimisation problems:
\begin{itemize}
    \item Given a 2-connected 4-regular planar graph, find a planar embedding in which the length of the longest threaded circuit is maximised.
    \item Given a 2-connected 4-regular planar graph, find a planar embedding in which the average of the lengths of threaded circuits in the threaded circuit partition is maximised.
\end{itemize}

\end{document}